\documentclass[12pt]{article}
\usepackage[T1]{fontenc}
\usepackage{booktabs}
\usepackage{amsfonts}
\usepackage{enumerate}
\usepackage{tabularx}
\usepackage{subfigure}
\usepackage{multirow}
\usepackage{ulem}
\usepackage{amsmath}
\usepackage{amsthm}
\usepackage{graphicx}

\usepackage{algorithm}
\usepackage{algorithmic}
\usepackage{indentfirst}
\usepackage{setspace}
\usepackage{color}

\def\R{{ \mathbb{R}}}

\def\E{{ \mathbb{E}}}

\def\arg{\mathop{\rm arg}}

\usepackage{natbib}
\bibpunct[, ]{(}{)}{,}{a}{}{,}%
\newtheorem{theorem}{\bf{Theorem}}

\newtheorem{lemma}{\bf{Lemma}}
\newtheorem{corol}{\bf{Corollary}}

\begin{document}




\title{Batch mode active learning for efficient parameter estimation}


\author{Wei Zheng\textsuperscript{1}, Ting Tian\textsuperscript{2}, and Xueqin Wang\textsuperscript{3}\\
\textsuperscript{1}Department of Business Analytics and Statistics, University of Tennessee\\
\textsuperscript{2}School of Mathematics, Sun Yat-sen University\\
\textsuperscript{3}School of Management, University of Science and Technology of China}
\date{}

\maketitle

\abstract{
For many tasks of data analysis, we may only have the information of the explanatory variable and the evaluation of the response values are quite expensive.  While it is impractical or too costly to obtain the responses of all units, a natural remedy is to judiciously select a good sample of units, for which the responses are to be evaluated. In this paper, we adopt the classical criteria in design of experiments to quantify the information of a given sample regarding parameter estimation. Then, we provide a theoretical justification for approximating the optimal sample problem by a continuous problem, for which fast algorithms can be further developed with the guarantee of global convergence. Our results have the following novelties: $(i)$ The statistical efficiency of any candidate sample can be evaluated without knowing the exact optimal sample; $(ii)$ It can be applied to a very wide class of statistical models; $(iii)$ It can be integrated with a broad class of information criteria; $(iv)$ It is much faster than existing algorithms. $(v)$ A geometric interpretation is adopted to theoretically justify the relaxation of the original combinatorial problem to continuous optimization problem.
}




%

\section{Introduction.}\label{intro} 
For many tasks of data analysis, we may only have the information of the explanatory variable while the response values are unavailable and expensive to obtain. As a result, it is impractical or too costly to obtain the responses of all units. With limited budget, we have to judiciously select a good sample of the units to obtain the response and analyse this sample in place of the full data. This problem has been extensively discussed in the machine learning community by the name of active learning. See \citet{settles2009active} and \citet{ren2021survey} for comprehensive surveys. Relatively, the study of sampling methods targeted on statistical inferences is quite limited in literature. \citet{deng2009active} studied the detection of money laundry activity by using logistic regression model to select the informative bank customer in a sequential procedure. In each step, one unit is selected from several pre-screened candidate units according to the D-criterion and added to the sample. The estimate of the parameters are correspondingly updated, based on which the next point will be selected. This sequential process may not be practical when the evaluation of the response is time consuming. Additionally, it fails to take the information overlap between the multiple data points into account.  In such situations, we may need to select a big batch of points simultaneously and the responses of the selected units are measured simultaneously, especially when a parallel system is available. In \citet{deng2009active}'s context, we shall confirm the money laundry activity for many bank customers simultaneously. Another disadvantage of the sequential approach is that it does not lead to the global optimization due to information overlap among successively selected units.

The problem is formally stated as follows. Let ${\cal X}=\{x_1,x_2,...,x_N\}$ be a set of all $N$ design points, and our task is to select a sample $S\subset \{1,2,...,N\}$ of a given size $|S|=n<N$ to obtain the response $\{y_i: i\in S\}$. Then we shall use the data $\{(x_i,y_i):i\in S\}$ to train a parametric model $p(y|x,\theta)$, where $\theta \in \R^k$ is the parameter of the distribution. The information matrix for $\theta$ based on $x$ is $M_x=-\E \partial^2 log(p(y|x,\theta))/\partial \theta^2\in \R^{k\times k}$. When the responses are independent and homoscedastic condition on the explanatory variable, the information matrix of the sample $S$ is simply $M(S)=\sum_{i\in S}M_{x_i}$. That means the asymptotic covariance matrix of the maximum likelihood estimator (MLE), $\hat{\theta}$, is proportional to $M(S)^{-1}$. For a fixed sample size, say $n$, we would like to choose the optimal sample which maximizes $M(S)$ in some sense. Suppose we are interested in $g(\theta)\in \R^q$, $1\leq q\leq k$, where $g$ is a differential function, the asymptotic covariance matrix of its MLE, $g(\hat{\theta})$, is given by
\begin{eqnarray*}
\Sigma_g(S)&=&\frac{\partial g(\theta)}{\partial \theta^T}M(S)^{-1}\left(\frac{\partial g(\theta)}{\partial \theta^T}\right)^T.
\end{eqnarray*}
In the design context, the covariance matrix $\Sigma_g(S)$ can be minimized under different criteria. For example, the criterion functions $tr(\Sigma_g(S))$, $|\Sigma_g(S)|^{1/q}$ and $\lambda_{max}(\Sigma_g(S))$ corresponds to the A, D and E criteria, respectively. Here, $tr(\cdot)$ denotes the trace of a matrix; $|\cdot|$ denotes the determinant of a matrix; and $\lambda_{max}(\cdot)$ denotes the largest eigenvalue of a matrix. In an effort to unify these criteria, \citet{kiefer1974general} defined the $\Phi_p$ function:
\begin{eqnarray}\label{eqn:1253}
\Phi_p(\Sigma_g(S))=\left(q^{-1}Tr\left[\Sigma_g(S)\right]^p\right)^{1/p},&&0\leq p<\infty.
\end{eqnarray}
A sample is said to be $\Phi_p$-optimal if it solves the following problem
\begin{eqnarray}\label{eqn:821}
\min_{S:|S|=n}\Phi_p(\Sigma_g(S)).
\end{eqnarray}
By direct calculations, one can show that
\begin{eqnarray*}
\lim_{p\rightarrow 0}\Phi_p(\Sigma_g(S))&=&|\Sigma_g(S)|^{1/q}.
\end{eqnarray*}
Hence, the $\Phi_p$ criterion reduces to D-criterion when $p=0$. Similarly, it reduces to the A and E criteria when $p$ takes the values of $1$ and $\infty$, respectively.

The combinatorial problem (\ref{eqn:821}) is generally N-P hard, and the aim of this paper is to develop a method of deriving a near-optimal solution in a fast manner. The method shall also permit the evaluation of the lower bound of the efficiency of any candidate solution. The rest of the paper is organized as follows. Section \ref{sec:geo} relaxes (\ref{eqn:821}) to a continuous optimization problem and provides a geometric interpretation of the connection between the two problems under the D-criterion. Moreover, two forms of optimality conditions are provided. Correspondingly, an algorithm is proposed for D-criterion. Section \ref{phi_p} generalizes the algorithm to the general $\Phi_p$-criterion and proves the convergence of the algorithm. Section \ref{sec:example} provides some numerical examples to illustrate that it runs a lot faster than existing algorithms.  

\section{Literature review}
Given a large pool of unlabeled unit, active learning provides a way to iteratively select the most informative unlabeled unites for the queries to label. Most work on active learning have focused on selecting a single unlabeled unit at a time. Among them, one main category is to measure the informativeness of each unlabeled point. \citet{zhang2000value} selects the unit that provides the most information for the prediction of unlabelled data. \citet{houlsby2011bayesian} selects points that maximally decreases the expected posterior entropy of the model parameters. They achieve this by adopting \citet{sebastiani2000maximum}'s formulation on the equivalence of entropy changes on the responses and parameters when a new unit is added to the data. This is helpful to handle high dimensional parameters and heteroscedastic cases such as the Guassian process. For other ideas, \citet{lewis1994sequential} selects the unit with conditional probability closest to $1/2$. \citet{freund1997selective} selects the unit on which a committee of classifiers disagree the most. \citet{campbell2000query} and \citet{tong2001support} suggest choosing the unit closest to the classification boundary and the latter formulated it as a version space reduction process. \citet{scheffer2001active}, \citet{wang2011active}, and \citet{chen2013near} used the difference in prediction probability of the most and second most probably class as the criterion.

The single unit selection demands for less computational resource, but also has many drawbacks. It requires retraining with all units when each new unit is added to the data. This is inefficient when the training is time consuming. Since the next unit selection has to wait for the labeling of the previous selected unit and this is a waste of resource and time when the parallel labeling system is available. Also, there will be certain information overlap among sequentially selected units since similar points that are ranked as most informative or most plausible by other criteria will be selected. To address this issue, some has proposed to select the subset of points simultaneously as presentative examples for the whole population. One natural
approach is based on clustering (\citet{li2012active,ienco2013clustering}). The idea is that after clustering is performed, the center of the clusters will be good representative examples of the points in the cluster. \citet{sener2017active} investigated the Core-set selection problem. The representative method is more diversified than the information based method, but it tends to select abundant points in high density area of the distribution.

Many have studied batch mode active learning (BMAL) to achieve the informativeness and diversification goal jointly. \citet{brinker2003incorporating}, \citet{schohn2000less}, and \citet{xu2003representative} extended single instance selection strategies that use support vector machines (SVM), and enforced diversity among points within the margin of SVM. \citet{guo2007discriminative} focused on the semisupervised learning, where the objective function for the optimal batch is the likelihood function of labeled data minus the entropy of the unlabeled data expressed with the logistic regression as the working model. The optimization is achieved by reformulated the objective function as an integer programming problem. \citet{hoi2006batch,hoi2006large} choose multiple instances that efficiently reduce the Fisher information. Specifically, \citet{hoi2006large} extended \citet{zhang2000value}'s work based on logistic regression to the batch mode active learning and and took the criterion of the ratio of the information matrices between all unlabeled samples and the selected unlabeled sampled to be labeled. The ratio could be understood as the capability of using the fitted model to predict all unlabeled samples once the selected samples are labeled. They adopted the greedy algorithm based on the idea of submodular function to reach the $1-1/e$ efficiency. The algorithm iteratively optimize on each unit of the batch by fixing the remaining fixed. The submodularity has been extensively applied in literature of active learning. For example, \citet{wei2015submodularity} connected it to likelihood functions of two classifiers for the multistage BMAL. \citet{ravi2016experimental} studied BMAL for lasso, where the objective function is the D-optimality plus penalty term instead of the entropy. \citet{kirsch2019batchbald} extended \citet{houlsby2011bayesian} to the batch model mutual information measurement. \citet{deng2022query} studied semi-supervised clustering problem, where they select the most informative pair to query based on the entropy criterion. They have all applied the submodule function for finding the subset based on the criteria defined therein. For deep learning models, \citet{zhdanov2019diverse} used the K-means clustering algorithm for the batch selection. \citet{ghorbani2021data} used Shapley value to evaluate the importance of each data point. Specifically, a larger size pool of data is preselected with high Shapley values under the KNN classifier and a diversity algorithm is applied to further select the desired size of batch from the pool.

Overall, many approaches use a variety of heuristics to guide the instance selection process, where the selected batch should be informative about the classification model while being diverse enough so that their information overlap is minimized. For those formulating the BMAL as optimization problem that aims to learn a good classifier directly, either an integer programing method or the submodular algorithm with $1-1/e$ efficiency is employed to solve the optimization problem. This is due to the NP-hard nature of subset selection problem. Here, we adopt a new method which first relax the combinatorial problem to a continuous optimization problem and a round procedure is adopted to obtain the discrete solution. Mathematically, this formulation is related to the constrained optimal design of experiments historically. See \citet{Wynn1977,wynn1982optimum, fedorov1989optimal, muller1998design, sahm2001note, pronzato2004minimax, pronzato2006sequential, ucinski2015algorithm} for example. However, these methods are not applicable for the current BMAL problem. One main reason is because exchange algorithm has been majorly adopted for the optimization. This is not computational efficient in dealing with large database. Meanwhile, models considered in literature are mostly of regression type (e.g. logistic regression) so that the information matrix from the unit is of the form $xx^T$, where $x$ is a column vector and $^T$ is the transpose operator. Recently, the backward algorithm is proposed by \citet{ouyang2016design} when they search for the D-optimal sample for the logistic regression model. In \citet{ouyang2016designed}, they further applied the procedure to precision medicine. \citet{ouyang2017batch} extended the discussion to other optimality criteria. Unfortunately, the backward algorithm is quite slow especially when the population size increases. Besides the speed issue of these two methods, the global statistical efficiency of the derived samples remains unknown. Here, we have provided geometric interpretations of our algorithm to justify the empirical evidenced close to 100\% efficiency. 

To summarize, our results have the following novelties. $(i)$ The statistical efficiency of any candidate sample can be evaluated without knowing the exact optimal sample; $(ii)$ It can be applied to a very wide class of statistical models; $(iii)$ It can be integrated with a broad class of information criteria; $(iv)$ It is much faster than existing algorithms. $(v)$ A geometric interpretation is adopted to theoretically justify the relaxation of the original combinatorial problem to continuous optimization problem.

\section{The geometric point of view}\label{sec:geo}
In this section, we assume the information matrix at a single point $x$ is of the form $M_x=xx^T$. Also, we focus on D-criterion and assume all parameters are of interest, i.e. $p=0$ and $g(\theta)=\theta$. Note a sample, say $S$, could be represented by a vector $\delta=(\delta_1,\delta_2,...,\delta_N)$, where $\delta_i=1$ if $i\in S$ and $\delta_i=0$ otherwise. As a result, we have $M(S)=\sum_{i=1}^N\delta_i x_ix_i^T$ and problem (\ref{eqn:821}) reduces to
\begin{equation}\label{eqn:129}
\max_\delta \left|\sum_{i=1}^N\frac{\delta_i}{n} x_ix_i^T\right|,
\end{equation}
subject to $\delta\in \{0,1\}^N$ and $\sum_{i=1}^N\delta_i=n$. To bypass the computational bottleneck of (\ref{eqn:129}), it is natural to replace $\delta_i/n$ by a continuous weight $w_i$ as follows.
\begin{equation}
\max_{w\in \Omega} log|M(w)|. \tag{P1}
\end{equation}
where $M(w)=\sum^N_{i=1}w_ix_ix_i^T$ and $\Omega=\{w=\{w_1,w_2,...,w_N\}: w_i\geq 0, 1\leq i\leq N, \sum^N_{i=1}w_i=1\}$. In approximate design theory, a vector $w\in \Omega$ is called a measure, $M(w)$ is its information matrix, and the solution of (P1) is called a D-optimal measure. However, (P1) is not a good approximation of (\ref{eqn:129}) since the number of points with positive weights (related to $k$) is typically far less than $n$, i.e. the pre-specified number of points to be selected. As shown in Figure \ref{fig1}, only two points are selected by (P1) while $n=10$ may be needed. 
\begin{figure}[H]
\centering
\includegraphics[scale=0.5]{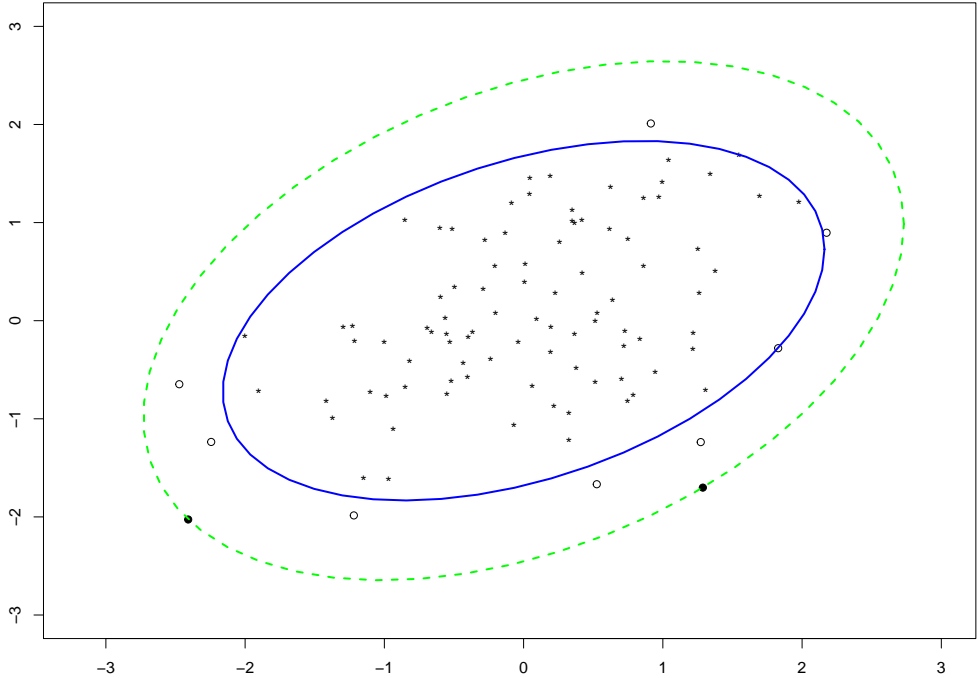}
\caption{The data consists of $100$ points generated from a two dimensional normal distribution. The two solid points are the supporting points of the solution of (P1) and the ten circled points are the supporting points of the solution of (P2) and (P1-$\epsilon$). The dashed ellipsoid is the solution of (D1) and the solid ellipsoid is the solution of (D2).}\label{fig1}
\end{figure}

To address this issue, we shall first review on the geometric view of this D-optimal design problem, and then explain our proposal in an intuitive way. By Section 3 of \citet{vandenberghe1998determinant}, the dual problem of (P1) is given by
\begin{equation}
\min_{W\in \R^{k\times k} }-log|W|,\tag{D1}
\end{equation}
subject to $W>0$ and $x_i^TWx_i\leq k$, $1\leq i\leq N$. Note that the set ${\cal E} (W)=\{x|x^TWx\leq k\}$ features an ellipsoid centered at origin and the region within it. Geometrically, (D1) tries to find the smallest volume ellipsoid among all ellipsoids which contain all points in ${\cal X}$ and are centered at the origin. See the dashed ellipsoid in Figure 1 for ${\cal E} (W^*)$, with $W^*$ being the solution of (D1).

Based on the duality between (P1) and (D1), we have $log|M(w)|\leq -log|W|$ with the equality holds if and only if $W^{-1}=M(w)$. As a result, we have the identity $W^*=M(w^*)^{-1}$, where $w^*=\{w_i^*,1\leq i\leq N\}$ is the solution of (P1). By Kiefer's general equivalence theorem for D-optimality, we have $w^*_j>0$ only for the points satisfying $x_j^TW^*x_j=k$. Geometrically, it means only the points on the boundary of ${\cal E} (W^*)$ could be the support point of a D-optimal measure. See the two solid circled points in the figure.

Note that the number of points on the boundary of ${\cal E} (W^*)$ is typically comparable to $k$ and hence far less than $n$. For the purpose of selecting $n$ points from ${\cal X}$, the geometric representation of (D1) inspires us to shrink the ellipsoid ${\cal E} (W^*)$ until $n$ points emerge on or outside the boundary of the new ellipsoid, as represented by the solid ellipsoid in Figure \ref{fig1}. The collection of these points will be a reasonable, though not necessarily optimal, solution for (\ref{eqn:129}). See the $n=10$ circled points outside the solid ellipsoid in the figure. To be more specific, we define the soft-margin minimum volume ellipsoid problem as
\begin{equation}
\min_{W\in \R^{p\times p},\xi_i\in\R }-log|W|+\lambda \sum^N_{i=1}\xi_i,\tag{D2}
\end{equation}
subject to $W>0$, $x_i^TWx_i\leq k+\xi_i$ and $\xi_i\geq 0$, $1\leq i\leq N$, with $\lambda$ being the penalty parameter. The introduction of the slack variable $\xi_i$ allows some points to fall outside the ellipsoid. Define the Lagrange $L=-log|W|+\lambda \sum^N_{i=1}\xi_i-\sum^N_{i=1}\alpha_i(k+\xi_i-x_i^TWx_i)-\sum^N_{i=1}\mu_i\xi_i$, from which we can derive the following Karush-Kuhn-Tucker (KKT) conditions.
\begin{eqnarray}
\frac{\partial L}{\partial W}=-W^{-1}+\sum^N_{i=1}\alpha_ix_ix_i^T&=&0,\label{eqn:124}\\
\frac{\partial L}{\partial \xi_i}=\lambda-\alpha_i-\mu_i&=&0,\label{eqn:1242}\\
\alpha_i,\mu_i&\geq& 0,  1\leq i\leq N,\label{eqn:125}\\
\mu_i\xi_i=\alpha_i(k+\xi_i-x_i^TWx_i)&=&0, 1\leq i\leq N.\label{eqn:1252}
\end{eqnarray}
Plugging (\ref{eqn:124}) and (\ref{eqn:1242}) back to the objective function of (D2), we derive its dual problem as
\begin{equation}
\max_{\alpha_i\in \R, 1\leq i\leq N}log \left|\sum^N_{i=1}\alpha_ix_ix_i^T\right| +k(1-\sum\alpha_i),\tag{P2}
\end{equation}
subject to $0\leq \alpha_i\leq \lambda$. Here we utilize the equation
\begin{eqnarray*}
\sum^N_{i=1}\alpha_ix_i^T\left(\sum^N_{i=1}\alpha_ix_ix_i^T\right)^{-1}x_i&=&k.
\end{eqnarray*}
As a result, the maximum of the objective function in (P2) shall be equal to the minimum of the objective function in (D2). For the solutions of (P2) and (D2), say $\alpha$ and $W$, we shall have the identity $W=M(\alpha)^{-1}$, where $M(\alpha)=\sum^N_{i=1}\alpha_ix_ix_i^T$.

To this end, we can characterize the solution of (P2) in Theorem \ref{thm:1015} below in view of (\ref{eqn:125}), (\ref{eqn:1252}) and the fact that $\xi_i>0$ if and only if $x_i^TWx_i>k$.
\begin{theorem}\label{thm:1015}
Suppose $\alpha=\{\alpha_1,...,\alpha_N\}$ is the solution of (P2). Then, for $1\leq i\leq N$, we have $(i)$ $\alpha_i=0$ if $x_i^TM(\alpha)^{-1}x_i<k$. $(ii)$ $\alpha_i=\lambda$ if $x_i^TM(\alpha)^{-1}x_i>k$. $(iii)$ $0\leq \alpha_i\leq \lambda$ if $x_i^TM(\alpha)^{-1}x_i=k$.
\end{theorem}

It is worth noticing that the form of (P2) appear close to the objective function in the literature of finite population sampling and constraint space optimal design of experiment. See \citet{Wynn1977,wynn1982optimum, fedorov1989optimal, muller1998design, sahm2001note, pronzato2004minimax, pronzato2006sequential, ucinski2015algorithm} for relevant literature. One particular version of such problems can be written as follows. 

\begin{equation}
\max_{w\in \Omega_{\epsilon}} log |M(w)|, \tag{P1-$\epsilon$}
\end{equation}
where $0<\epsilon<1$ and $\Omega_{\epsilon}=\{(w_1,w_2,...,w_N)\in \Omega: 0\leq w_i\leq \epsilon, 1\leq i\leq N\}$. Here, we shall establish the connection between (P2) and (P1-$\epsilon$), and then provide further insights, improvements and extensions that are vital for our current context of BMAL. 

For the connection, we shall prove that the solution path of (P2) is covered by (P1-$\epsilon$) as illustrated as in Theorem \ref{thm:p2p1}. Relatively speaking, the objective function in (P1-$\epsilon$) is simpler than that in (P2). By choosing a proper value of $\epsilon$, e.g. $\epsilon=1/n$, we can guarantee at least $n$ positive weights in the solution of (P1-$\epsilon$) and the sample could be chosen according to the $n$ largest weights. Since our goal here is to select the important points, it is not the absolute but the relative values of $\alpha_i$ in (P2) that matters. 
\begin{theorem}\label{thm:p2p1}
For given solution of (P2), say $\alpha=\{\alpha_1,...,\alpha_N\}$, under the penalty $\lambda$. Let $\epsilon=\lambda/\sum^N_{i=1}\alpha_i$, the solution of (P1-$\epsilon$) will be given by $w=\alpha/\sum^N_{i=1}\alpha_i$.
\end{theorem}
\begin{proof}
Let $w_i=\alpha_i/\sum^N_{i=1}\alpha_i$, the target function of (P2) becomes $log|M(w)|+k(1-\sum^N_{i=1}\alpha_i+\log(\sum^N_{i=1}\alpha_i))$. By fixing the second term at the solution of (P1-$\epsilon$), the maximization of the first term can be achieved by solving (P1-$\epsilon$) under the constraint of $0\leq w_i=\alpha_i/\sum^N_{i=1}\alpha_i\leq \lambda/\sum^N_{i=1}\alpha_i=\epsilon$. The constraint of $\sum^N_{i=1}w_i=1$ is inherently guaranteed by the construction of $w_i$.
\end{proof}

To this end, we have justified approximating problem (\ref{eqn:129}) by (P1-$\epsilon$). While this latter problem has been studied in literature, there are several questions and issues as listed below that are not resolved for the current application. It is main purpose of this paper to address these issues. 
\begin{enumerate}
\item The upper bound $\epsilon$ serves to enforce at least $1/\epsilon$ points to have positive weights, but it does not necessarily means the points such selected is a good solution for the original problem (\ref{eqn:821}). Do we have a theoretical justification for (P1-$\epsilon$)?
\item There seems to be no control for the maximum number of points with positive weights. What if there are too many of them as compared to $n$?
\item So far, the algorithms proposed in literature are statistically efficient, but not computational efficient. Do we have a much faster algorithm that still achieves the high statistical efficiency?
\item So far, most work typically focus on a particular model and work on a particular criterion, most frequently D-criterion. Can we propose an algorithm that is readily applicable to any arbitrary model with interest in any function ($g$) of the parameters under any $\Phi_p$-criterion?
\end{enumerate} 

The first issue has been largely addressed by the primal-dual argument coupled with the geometric interpretation of the problem as illustrated by Figure \ref{fig1}. Through the empirical results, we shall observe that our proposed algorithm would achieve close 100\% efficiency for all examples. The second issue is addressed by Corollary \ref{corol:trich}, which indicates that the number of points with positive weights will be very close to $n$ when we set $\epsilon=1/n$. The solutions for the last two issues will be provided in Section \ref{phi_p}.

\begin{corol}\label{corol:trich}
Suppose $w=\{w_1,...,w_N\}$ is the solution of (P1-$\epsilon$). Then there exists a constant, say $c$, such that for $1\leq i\leq N$, $(i)$ $w_i=0$ if $x_i^TM(w)^{-1}x_i<c$. $(ii)$ $w_i=\epsilon$ if $x_i^TM(w)^{-1}x_i>c$. $(iii)$ $0\leq w_i\leq \epsilon$ if $x_i^TM(w)^{-1}x_i=c$.
\end{corol}

This corollary can be obtained by investigating the linear proportionality between the solutions of (P2) and (P1-$\epsilon$). It characterized the solution of (P1-$\epsilon$) by the trichotomy structure as in Theorem \ref{thm:1015}, which is the key for the development of efficient algorithms. Note that the proof of Corollary \ref{corol:trich} is, to some extend, inspired by the KKT condition arguments adopted in supporter vector machine. As shown in Figure \ref{fig_lev}, only few points will have the weight strictly between $0$ and $\epsilon$. In fact, these are the points fall exactly on the small ellipsoid in Figure 1. Similarly, the points with $w_i=\epsilon$ is exactly the points outside the small ellipsoid in Figure 1. This addresses the second issue as mentioned earlier. The short answer is that the total number of points with positive weight will be only slightly larger than $n$, if not equal. 
\begin{figure}[H]
\centering
\includegraphics[scale=0.5]{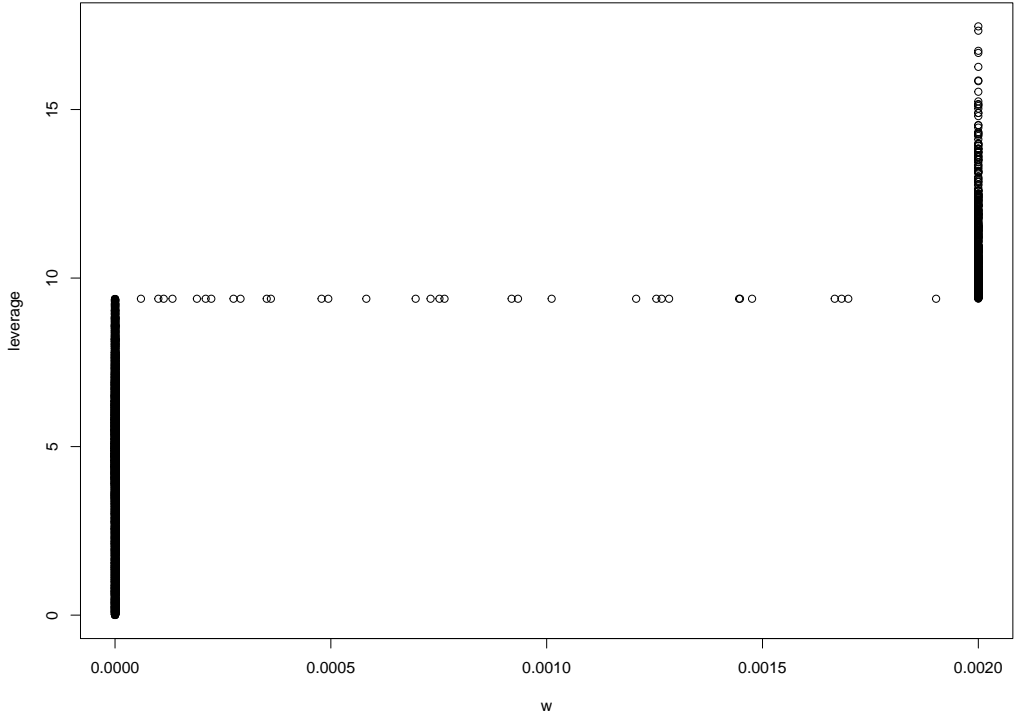}
\caption{The plot of leverage score $x_i^TM(w)^{-1}x_i$ versus the weight $w_i$ when $w$ is optimal for (P1-$\epsilon$).}\label{fig_lev}
\end{figure}


To address the third issue, Theorem \ref{thm:eqv} provide a Kiefer's type of equivalence theorem for the optimality conditions of (P1-$\epsilon$), which sheds light on the construction of Algorithm 1 below. Once the solution of (P1-$\epsilon$) is derived, we shall project it to the solution space of (\ref{eqn:129}) as a reasonable solution for (\ref{eqn:129}). As a result, our sample consists of the $n$ points with the largest weights in the solution of (P1-$\epsilon$).

\begin{theorem}\label{thm:eqv}
(i) A measure $w$ solves (P1-$\epsilon$) if and only if
\begin{eqnarray}\label{eqn:1293}
\max_{w'\in \Omega_{\epsilon}} Tr(M(w)^{-1}M(w'))&=&k
\end{eqnarray}
with the maximum achieved by $w'=w$. \\
(ii) If a measure $w$ is a solution of (P1-$\epsilon$), we have $w_i\leq w_j$ whenever $x_i^TM(w)^{-1}x_i<x_j^TM(w)^{-1}x_j$.
\end{theorem}
\begin{proof}
Let $L(w)=log|M(w)|$. By the concavity of $L(w)$, $w$ maximises $L(w)$ if and only if
\begin{eqnarray}\label{eqn:1296}
0\geq  \left.\frac{\partial L((1-\alpha)w+\alpha w')}{\partial \alpha}\right|_{\alpha=0}=Tr(M(w)^{-1}M(w'))-k
\end{eqnarray}
for any measure $w'\in \Omega_{\epsilon}$. That indicates
\begin{eqnarray}\label{eqn:1294}
\max_{w'\in \Omega_{\epsilon}} Tr(M(w)^{-1}M(w'))&\leq& k
\end{eqnarray}
Since $Tr(M(w)^{-1}M(w))=k$ and
\begin{eqnarray}\label{eqn:1297}
Tr(M(w)^{-1}M(w'))=\sum^N_{i=1}w'_ix_i^TM(w)^{-1}x_i
\end{eqnarray}
is linear in $w'$, hence we should also have
\begin{eqnarray}\label{eqn:1295}
\max_{w'\in \Omega_{\epsilon}} Tr(M(w)^{-1}M(w'))&\geq &k
\end{eqnarray}
Then (\ref{eqn:1293}) is a direct result of (\ref{eqn:1294}) and (\ref{eqn:1295}). For part $(ii)$, suppose $w_i> w_j$, (\ref{eqn:1296}) shall be violated by taking $w'_i=w_j$, $w'_j=w_i$ and the rest components of $w'$ being equal to that of $w$.
\end{proof}

To introduce our algorithm for (P1-$\epsilon$), we define a measure, say $\tilde{w}$, to be the {\it steepest gradient} (SG) measure of $w$ if it achieves the maximum in (\ref{eqn:1293}). For a given measure, its SG measure represents the direction of fastest improvement on the objective function. Particularly, a measure is the SG measure of itself if it is the optimal solution of (P1-$\epsilon$).  Let $S_n(w)\subset \{1,2,...,N\}$ be the index of $n$ points with the largest value of $x_i^TM(w)^{-1}x_i$. By (\ref{eqn:1297}), we shall be able to explicitly construct the SG measure of $w$ in the following way. If $1/\epsilon$ is an integer, let $\tilde{w}_i=\epsilon$ for $i\in S_{1/\epsilon}(w)$ and $\tilde{w}_i=0$ for $i\not\in S_{1/\epsilon}(w)$. In general, set $\tilde{w}_i=\epsilon$ for $i\in S_{[1/\epsilon]}(w)$, $\tilde{w}_i=1-\epsilon[1/\epsilon]$ for $i\in S_{[1/\epsilon]+1}(w)-S_{[1/\epsilon]}(w)$ and $\tilde{w}_i=0$ for $i\not\in S_{[1/\epsilon]+1}(w)$. Here, $[1/\epsilon]$ is the largest integer that is smaller than $1/\epsilon$. The value $x_i^TM(w)^{-1}x_i$ is called the leverage score in literature. It is the variance of predicted value of the response at $x_i$ in a linear model. Hence, a SG measure tries to collect the most unpredictable points from the dataset. Now we propose the following algorithm for (P1-$\epsilon$).

\begin{center}
\fbox{
  \parbox{\textwidth}{
{\bf Algorithm 1}

{\bf Step 1:} Initialization. Give the initial solution $w^0$ such that $M(w^0)>0$.

{\bf Step 2:} Update. At iteration $t\geq 1$, calculate the SG measure of $w^{t-1}$, namely $\tilde{w}^{t-1}=\{\tilde{w}_1^{t-1},...,\tilde{w}_N^{t-1}\}$. Let $T^t_1=\{i:w^{t-1}_i=\tilde{w}^{t-1}_i=\epsilon\}$ and $T^t_2=\{i:w^{t-1}_i=\tilde{w}^{t-1}_i=0\}$. Maximize $log|M(w)|$ subject to the constraints that $w_i=\epsilon$ for $i\in T^t_1$, $w_i=0$ for $i\in T^t_2$ and $w\in \Omega_{\epsilon}$. Denote the solution by $w^t$.

{\bf Step 3:} Check the optimality based on Theorem \ref{thm:eqv}. If $Tr(M(w^t)^{-1}M(\tilde{w}^t))<k(1+v)$, with $v$ being a sufficiently small positive constant, stop and return the value of $w^t$. Otherwise, repeat Step 2 until the convergence criterion is met.

  }
}
\end{center}  

The computational complexity of deriving the SG measure is $O(k^3+Nk^2)$ and the computational complexity of maximizing $log|M(w)|$ step 2 is $O((m+n)k^2+k^3)$, where $m$ is the number of elements neither in $T^t_1$ nor in $T^t_2$. Hence it is important to choose a proper initial solution so that $m$ is small in the beginning. This can be achieved by selecting the set $S$ of $n$ points according to the largest leverage score $x_i^T(\sum^N_{i=1}x_ix_i^T)^{-1}x_i$ and let $w_i^0=1/n$ for $i\in S$. It is justified by the fact that the leverage score is proportional to the derivative of $log|M(w)|$ with respective to $w_i$ evaluated at the uniform measure $w_i=1/N$, $1\leq i\leq N$. Hence, the initial measure such selected is essentially the SG measure of the uniform measure. As a result, we have $m\leq 2n$ and typically we have $m=o(n)$ especially in later iterations. To sum up, the computational complexity of Algorithm 1 is $O(k^3+Nk^2)$. For fixed dimension $k$, it is simply $O(N)$. Another idea to expedite the algorithm is to update the measure by a linear combination of the measure and its SG measure with the weight carefully chosen. These ideas will be implemented in Algorithm 3 for the general $\Phi_p$ optimality.

\section{The $\Phi_p$ optimality}\label{phi_p}
In this section, we try to address the last two issues as proposed in Section \ref{sec:geo}, that is to derive an efficient (both statistically and computationally) algorithm to derive near-optimal solution of (\ref{eqn:821}) in a unified framework for any model, any function of parameters and any design criterion. Our framework shall permit the evaluation of the statistical efficiency of the derived samples as compared to the actual optimal one, even thought the latter may not be known. The theoretical results regarding the global convergence is also obtained. The speed of the algorithm will be illustrated by the examples in the next section.

Inspired by the geometric justification of (P1-$\epsilon$) for D-criterion, we propose to approximate (\ref{eqn:821}) by the following continuous optimization problem.
\begin{eqnarray}\label{eqn:130}
\min_{w\in \Omega_\epsilon}\Phi_p(\Sigma_g(w)),
\end{eqnarray}
for $0\leq p<\infty$, where $\Phi_p$ is defined in (\ref{eqn:1253}) and
\begin{eqnarray*}
\Sigma_g(w)&=&\frac{\partial g(\theta)}{\partial \theta^T}M(w)^{-1}\left(\frac{\partial g(\theta)}{\partial \theta^T}\right)^T,\\
M(w)&=&\sum^N_{i=1}w_iM_{x_i}.
\end{eqnarray*}
Here, the information matrix $M_{x_i}$ from a single point $x_i$ does not has to be of the particular form $x_ix_i^T$. When we take $g$ to be the identity function and $p=0$, we have $\Phi_p(\Sigma_g(w))=|M(w)|^{-1/k}$. Hence, (\ref{eqn:130}) covers (P1-$\epsilon$) as a special case. For a sample $S$ of size $n$, let $w(S)\in \Omega$ be the measure assigning weight $1/n$ to each element of $S$. Obviously, we have $\Phi_p(\Sigma_g(w(S)))=n\Phi_p(\Sigma_g(S))$. Note that $w(S)\in \Omega_\epsilon$ for all $S\subset \{1,2,...,N\}$ if and only if $\epsilon\geq 1/n$. In this case, the domain of (\ref{eqn:130}) covers the domain of (\ref{eqn:1253}) and hence we can evaluate the (lower bound of) efficiency of $S$
\begin{eqnarray}\label{eqn:207}
{\cal E}_p(w)&:=&\frac{\Phi_p(\Sigma_g(w^*))}{\Phi_p(\Sigma_g(w(S)))},
\end{eqnarray}
where $w^*$ is the solution of (\ref{eqn:130}). As in Section \ref{sec:geo}, we shall sample the $n$ points with the largest weights in $w^*$ as an approximate solution for the minimization of (\ref{eqn:1253}). To make sure $w^*$ has at least $n$ positive weights, we need the inequality $\epsilon\leq 1/n$. As a result, we would suggest to simply set $\epsilon=1/n$ in general.

Now we shall generalize the results of Section \ref{sec:geo} for the $\Phi_p$ criterion and prove the global convergence of the generalized algorithm. The following result plays a crucial role in the development of Algorithms 2 to be proposed.

\begin{lemma}\label{lemma1}
Let $\Omega_{\epsilon}^+=\{w\in \Omega_\epsilon: M(w)>0\}$. For $0\leq p<\infty$, $\Phi_p(\Sigma_g(w))$ is convex on the domain of $\Omega_{\epsilon}^+$.
\end{lemma}

The proof follows by the same arguments as in Lemma 1 of \citet{yang2013optimal} under the classical framework of approximate design theory, except that $\xi_0$ therein shall be deleted. In \citet{yang2013optimal}, the inclusion of $\xi_0$ as the first stage design is necessary for the proof of the convexity since they worked on the domain $\Omega$ for the second stage design. Here, we bypass the need for the first stage design by shrinking the domain from $\Omega_{\epsilon}$ to $\Omega_{\epsilon}^+$. Such adjustment is crucial for the establishment of the theoretical convergence of our proposed algorithms. By Lemma \ref{lemma1}, a measure $w$ is a solution of (\ref{eqn:130}) if and only if the Fr\'echet direvitive of the objective function is nonnegative, namely
\begin{eqnarray}\label{eqn:1302}
\eta(w',w):= \left.\frac{\partial \Phi_p(\Sigma_g((1-\alpha)w+\alpha w'))}{\partial \alpha}\right|_{\alpha=0}&\geq &0.
\end{eqnarray}
By direct calculation we have
\begin{eqnarray}
-\eta(w',w)&=&\sum^N_{i=1} (w'_i-w_i)\phi(x_i,w)\label{eqn:1303}\\
&=&\sum^N_{i=1} w'_i\phi_p(x_i,w)-\Phi_p(\Sigma_g(w)),\nonumber
\end{eqnarray}
\begin{eqnarray}
\phi_p(x,w)&=&q^{-1/p}\left(Tr\left[\Sigma_g(w)\right]^p\right)^{1/p-1}\nonumber\\
&&\times Tr\left(\Sigma_g(w)^{p-1}\frac{\partial g(\theta)}{\partial \theta^T}M(w)^{-1}M_xM(w)^{-1}\left(\frac{\partial g(\theta)}{\partial \theta^T}\right)^T\right).\nonumber
\end{eqnarray}

The value of $\phi_p(x,w)$ can be considered as the $\Phi_p$ version of the leverage. By using the same arguments as in Theorem \ref{thm:eqv} with (\ref{eqn:1296}) and (\ref{eqn:1297}) therein replaced by  (\ref{eqn:1302}) and (\ref{eqn:1303}), we have

\begin{theorem}\label{thm:geqv}
(i) A measure $w$ solves (\ref{eqn:130}) if and only if
\begin{eqnarray}
\max_{w'\in \Omega_{\epsilon}} \sum^N_{i=1} w'_i\phi_p(x_i,w)&=&\Phi_p(\Sigma_g(w))
\end{eqnarray}
with the maximum achieved by $w'=w$. \\
(ii) If a measure $w$ is a solution of (P1-$\epsilon$), we have $w_i\geq w_j$ whenever $\phi_p(x_i,w)>\phi_p(x_j,w)$.
\end{theorem}

Theorem \ref{thm:geqv} generalizes Theorem \ref{thm:eqv} in characterising the solution of (\ref{eqn:130}). Moreover, the trichotomy phenomena for (P1-$\epsilon$) also exists for (\ref{eqn:130}) with the details given by Corollary \ref{cor2}.

\begin{corol}\label{cor2}
Suppose $w=\{w_1,...,w_N\}$ is the solution of (\ref{eqn:130}). Then there exists a constant, say $c$, such that for $1\leq i\leq N$, $(i)$ $w_i=0$ if $\phi_p(x_i,w)<c$. $(ii)$ $w_i=\epsilon$ if $\phi_p(x_i,w)>c$. $(iii)$ $0\leq w_i\leq \epsilon$ if $\phi_p(x_i,w)=c$.
\end{corol}

To develop the algorithm for (\ref{eqn:130}), we shall generalize the definition of SG measure $\tilde{w}$ of $w$ in Section \ref{sec:geo} to be the measure satisfying the following equation.
\begin{eqnarray}\label{eqn:1313}
\eta(\tilde{w},w)&=&\min_{w'\in\Omega_\epsilon}\eta(w',w).
\end{eqnarray}
The convergence of an algorithm typically relies on the convexity of the objective function. Inspired by Lemma \ref{lemma1}, we shall consider the measure which achieves the minimum in (\ref{eqn:1313}) among $\Omega_\epsilon^+$ instead of $\Omega_\epsilon$. Such a measure is said to be the {\it positive steepest gradient} (PSG) measure of $w$. Since $n>>k$, the SG measure is almost always the PSG measure in practice. For data generated by continuous distribution, even $n=k$ will insure the positiveness of the SG measure with probability $1$. We use the PSG measure in Algorithm 2, only for the need in the proof of convergence in Theorem \ref{thm:conv}. In practice, we can directly adopt the SG measure while its positiveness is checked during the process. The construction of SG measure for the $\Phi_p$ criterion is similar to that for the D criterion, only by replacing $x_i^TM(w)^{-1}x_i$ therein by $\phi_p(x_i,w)$. Particularly, when $1/\epsilon$ is an integer, the SG measure of $w$ is the measure with its weight equal to $\epsilon$ for $1/\epsilon$ points with largest values of $\phi_p(x_i,w)$.
\begin{center}
\fbox{
  \parbox{\textwidth}{
{\bf Algorithm 2}

{\bf Step 1:} Initialization. Give the initial solution $w^0$ such that $M(w^0)>0$.

{\bf Step 2:} Update. At iteration $t\geq 1$, calculate the PSG measure of $w^{t-1}$, namely $\tilde{w}^{t-1}=\{\tilde{w}_1^{t-1},...,\tilde{w}_N^{t-1}\}$. Minimize $\Phi_p(\Sigma_g(w))$ subject to the constraints that $w_i=\epsilon$ for $i\in T^t_1$, $w_i=0$ for $i\in T^t_2$ and $w\in \Omega_{\epsilon}$, where $T^t_1$ and $T^t_2$ are defined as in Algorithm1. Denote the solution by $w^t$.

{\bf Step 3:} Check the optimality based on Theorem \ref{thm:geqv}. If $\sum^N_{i=1} \tilde{w}_i^t\phi_p(x_i,w)>(1-v)\Phi_p(\Sigma_g(w^t))$, with $v$ being a sufficiently small positive constant, stop and return the value of $w^t$. Otherwise, repeat Step 2 until the convergence criterion is met.

  }
}
\end{center}

The minimization in Step 2 can be carried out by any generic optimization method, where the derivative of the objective function with respect to $w_i$ is $\phi_p(x_i,w)$ and the domain is $\Omega_\epsilon$. Theorem \ref{thm:conv} proves the global convergence of Algorithm 2.

\begin{theorem}\label{thm:conv}
Let $w^t$ be the solution at iteration $t$ of Algorithm 2 and $\hat{w}$ be the solution of (\ref{eqn:130}), we have
\begin{eqnarray}\label{eqn:131}
\lim_{t\rightarrow \infty}\Phi_p(\Sigma_g(w^t))&=&\Phi_p(\Sigma_g(\hat{w})).
\end{eqnarray}
\end{theorem}
\begin{proof}
Suppose (\ref{eqn:131}) does not hold, we shall have a constant, say $\Delta>0$, such that
\begin{eqnarray}\label{eqn:1312}
\Phi_p(\Sigma_g(w^t))-\Phi_p(\Sigma_g(\hat{w}))&\geq&\Delta,
\end{eqnarray}
for all $ t\geq 1$. By Lemma \ref{lemma1}, we have
\begin{eqnarray*}
\Phi_p(\Sigma_g((1-\alpha)w^t+\alpha\hat{w}))&\leq&(1-\alpha)\Phi_p(\Sigma_g(w^t))+\alpha\Phi_p(\Sigma_g(\hat{w})),
\end{eqnarray*}
for any $0\leq \alpha\leq 1$. By (\ref{eqn:1302}) and simple algebra and calculus, we have
\begin{eqnarray}\label{eqn:1314}
\eta(\hat{w},w^t)&\leq&\Phi_p(\Sigma_g(\hat{w}))-\Phi_p(\Sigma_g(w^t)),
\end{eqnarray}
Let $\tilde{w}^t$ be the PSG measure of $w^t$. By (\ref{eqn:1313}), (\ref{eqn:1312}) and (\ref{eqn:1314}), we have
\begin{eqnarray}\label{eqn:1315}
\eta(\tilde{w}^t,w^t)&\leq&\eta(\hat{w},w^t)\leq -\Delta
\end{eqnarray}
Denote
\begin{eqnarray}\label{eqn:629}
\tau(\alpha,w',w)&:=& \frac{\partial^2 \Phi_p(\Sigma_g((1-\alpha)w+\alpha w'))}{\partial \alpha^2}
\end{eqnarray}
By the same arguments for Theorem 3 of \citet{yang2013optimal}, there exist a constant $K$ such that $\sup\{\tau(\alpha,\tilde{w}^t,w^t):t\geq 1, 0\leq \alpha\leq 1/2\}=K<\infty$. Now we have the taylor's expansion
\begin{eqnarray*}
\Phi_p(\Sigma_g((1-\alpha)w^t+\alpha\tilde{w}^t))&=&\Phi_p(\Sigma_g(w^t))+\alpha\eta(\tilde{w}^t,w^t)+\alpha^2\tau(\alpha',\tilde{w}^t,w^t)/2\\
&\leq&\Phi_p(\Sigma_g(w^t))-\Delta\alpha+K\alpha^2/2
\end{eqnarray*}
where $0\leq \alpha'\leq \alpha\leq 1/2$. By taking $\alpha=\Delta/K\leq 1/2$, we have
\begin{eqnarray}\label{eqn:1316}
\Phi_p(\Sigma_g((1-\alpha)w^t+\alpha\tilde{w}^t))&\leq&\Phi_p(\Sigma_g(w^t))-\Delta^2/(2K).
\end{eqnarray}
However, the updating step of Algorithm 2 implies
\begin{eqnarray}\label{eqn:1317}
\Phi_p(\Sigma_g(w^{t+1}))\leq \Phi_p(\Sigma_g((1-\alpha)w^t+\alpha\tilde{w}^t)).
\end{eqnarray}
By (\ref{eqn:1316}) and (\ref{eqn:1317}), we have $\Phi_p(\Sigma_g(w^{t}))=\Phi_p(\Sigma_g(w^0))-t\Delta^2/(2K)\rightarrow -\infty$ as $t$ goes to infinity, which contradict with the fact that $\Phi_p(\Sigma_g(w))\geq 0$ for all $w\in \Omega$.
\end{proof}

Inspired by (\ref{eqn:1316}) and (\ref{eqn:1317}), we propose to speed up the updating step as follows. For $\alpha$ in the neighbourhood of zero, we have the approximation
\begin{eqnarray}
\Phi_p(\Sigma_g((1-\alpha)w+\alpha\tilde{w}))&\approx&\Phi_p(\Sigma_g(w))+\alpha\eta(\tilde{w},w)\label{eqn:6292}\\
&&+\alpha^2\tau(0,\tilde{w},w)/2.\nonumber
\end{eqnarray}
Here, the function $\tau$ as defined in (\ref{eqn:629}) is simply the second order derivative. By the convexity of $\Phi_p$, we have $\tau(0,\tilde{w},w)\geq 0$ with equality achieved by the optimal measure. Hence, the right hand side of (\ref{eqn:6292}) is minimized at $\alpha^*=-\eta(\tilde{w},w)/\tau(0,\tilde{w},w)$ and we shall update $w$ by the new measure as $(1-\alpha^*)w+\alpha^*\tilde{w}$. To avoid the singularity issue, we can add a very small amount to the denominator of $\alpha^*$. Since $\eta(\tilde{w}^{t-1},w^{t-1})\leq 0$ with the equality achieved by the optimal measure, we shall have $\alpha^*\geq0$ in this case. To avoid the negative case due to numerical error, we can truncate $\alpha^*$ by $0$ from the left. Meanwhile, in order for the approximation to be valid, we should truncate $\alpha^*$ in the right by a small value, say $0<r<1$. In practice, $r=1/4$ works well. To summarize, we shall take
\begin{eqnarray}\label{eqn:816}
\alpha^*&=&\min\left[r,\max\left(0,\frac{-\eta(\tilde{w},w)}{\tau(0,\tilde{w},w)+u}\right)\right]
\end{eqnarray}
with $u>0$ being an arbitrarily small constant. Due to its simplicity, this updating procedure $(1-\alpha^*)w+\alpha^*\tilde{w}$ shall be much faster than step 2 of Algorithm 2. Obviously, it can quickly boost up the speed of updating process in the early stages. However, it will be converge slowly in the later stages when precision requirement is high, that is when $v$ in step 3 of Algorithm 2 is small. The slow down phenomena could be explained by the sparsity property of the optimal solution as characterized by Corollary \ref{cor2}. Hence, we shall use the updating step of Algorithm 2 in the later stages. To take advantages of both sides, we propose the following hybrid algorithm.

\begin{center}
\fbox{
  \parbox{\textwidth}{
 {\bf Algorithm 3} (Hybrid Algorithm)

{\bf Step 1:} Initialization. Take initial solution $w^0$ to be the SG measure of the uniform measure.

{\bf Step 2:} Boost. At iteration $t\geq 1$, calculate the SG measure of $w^{t-1}$, namely $\tilde{w}^{t-1}$. Let $w^t=(1-\alpha^t)w^{t-1}+\alpha^t\tilde{w}^{t-1}$, where $\alpha^t$ is calculated by (\ref{eqn:816}).

{\bf Step 3:} Check. If $\sum^N_{i=1} \tilde{w}_i^t\phi_p(x_i,w)>(1-v_0)\Phi_p(\Sigma_g(w^t))$, with $v_0$ being a sufficiently small positive constant. If yes, go to step 4. Otherwise, repeat Step 2 until the criterion is met.

{\bf Step 4:} Optimization. Run Algorithm 2 with the solution from Step 3 as the initial solution and the precision parameter $v>0$.

  }
}
\end{center}

In general, we take $v<v_0$ so that step 4 will further improve on the measure provided by step 3. The flexibility of the algorithm also allows us to take $v>v_0$ so that step 4 will be skipped. This can be useful when the accuracy requirement is not super high. In practice, steps 1-3 of Algorithm 3 is very time effective for $v=0.001$ or larger. Since Algorithm 2 is embedded into the hybrid algorithm, we shall still have the global convergence of the latter. 

\begin{corol}
Let $w^t$ be the solution at iteration $t$ of the hybrid algorithm and $\hat{w}$ be the solution of (\ref{eqn:130}), we have
\begin{eqnarray}
\lim_{t\rightarrow \infty}\Phi_p(\Sigma_g(w^t))&=&\Phi_p(\Sigma_g(\hat{w})).
\end{eqnarray}
\end{corol}

\section{Examples}\label{sec:example}

\subsection{Logistic regression}\label{sec:simple}
Let $Y_i\in \{0,1\}$ and $z_i\in \R^k$, $1\leq i\leq N$, be the binary response and the exploratory variable, respectively. In logistic regression, the probability of success, $p_i=P(Y_i=1)$, is modelled by
\begin{eqnarray}\label{eqn:228}
log(p_i/(1-p_i))=z_i^T\beta.
\end{eqnarray}
The information matrix of a sample $S\subset \{1,2,...,N\}$ for $\beta$ is $M(S)=\sum_{i\in S}x_ix_i^T$, where $x_i=\sqrt{p_i(1-p_i)}z_i$. Correspondingly, we have $M(w)=\sum_{i=1}^Nw_ix_ix_i^T$ for a measure $w\in \Omega_\epsilon$. This model coupled with the objective function (\ref{eqn:129}) is adopted by \citet{ouyang2016design} when they tried to solve the validation sampling problem for error-prone medical records. They proposed the backward algorithm, that is to begin with $S^0=\{1,2,...,N\}$ and delete one point each time until the sample size reduces to $n$. At iteration $t$, they delete the point from $S^t$ with smallest value of $x_i^TM(S^t)x_i$, namely the point most predictable by the model fitted to $S^t$. This is a type of greedy algorithm, which tries to minimize the loss of information at each step of deletion. Besides the proposed algorithms in this paper, we will also include the exchange algorithm in the comparison. That is to iteratively delete the point with the smallest leverage in the sample and add the point with the largest leverage outside the sample in each updating step. We have argued that (\ref{eqn:129}) can be approximated by (\ref{eqn:207}) with $p=0$ and $g$ being the identity function. The latter can be solved by the hybrid algorithm proposed in this paper. We shall evaluate the (lower bound of) D-efficiency of a candidate sample by
\begin{eqnarray}
{\cal E}_0(w(S))&=&\frac{|M(w(S))|^{1/k}}{|M(w^*)|^{1/k}},
\end{eqnarray}
where $w^*$ is the solution of (\ref{eqn:207}).


An experiment to compare three algorithms in terms of both time and statistical efficiency is carried out in the R environment on a desktop with 3.4 GHz cpu. The study is based on simulated data, where each point consists of the constant term and a random vector from $10$ dimensional standard normal distribution. The results are given in Table \ref{table1}. While all the algorithms are able to produce near-optimal samples, the time efficiency turns out to be significantly different. The exchange algorithm only takes a fraction of the time needed by the backward algorithm and the hybrid algorithm further takes a fraction of what is needed by the exchange algorithm. It is possible to tune the exchange algorithm to run faster, but it unlikely the improved version could outperform the hybrid algorithm. We can also notice in the table that the D-efficiency of the backward algorithm is slightly higher than the other two by a tiny margin, that is because exchange and hybrid algorithms are set to stop at $v=10^{-3}$ for $N=10,000$ and $v=10^{-2}$ for larger values of $N$. If we set $v=10^{-6}$ for $N=10,000$, the hybrid algorithm produces a sample with D-efficiency of $0.999995$ and it takes only $3.33$ seconds.

\begin{table}[h]

\begin{center}
\begin{tabular}{|c|c|c|c|c|}
\hline
N&10,000&100,000&1,000,000&10,000,000\\ \hline
\multirow{2}{*}{Backward} & 6.64 mins&12.88 hours&NA&NA\\
                          & 0.9999939&0.9999895&NA&NA\\ \hline
\multirow{2}{*}{Exchange} &  24.02 secs  & 5.31 mins &1.24 hours &20.40 hours \\
                          &  0.9994489  & 0.9950084 & 0.9956839& 0.9947914 \\\hline
\multirow{2}{*}{Hybrid}&0.68 secs& 5.75 secs&1.17 mins&20.38 mins\\
                       &0.9999827& 0.999884&0.999649&0.9993768\\ \hline
\end{tabular}
\end{center}
\caption{The time cost and D-efficiencies of the backward, exchange and hybrid algorithms for selecting $n=1000$ points from populations of different sizes. The dimension is $k=11$. Each point consists of the constant term and a random vector from $10$ dimensional standard normal distribution.}\label{table1}
\label{default}
\end{table}

While the backward algorithm works on the D-optimality and the case when all the parameters of are interest, the hybrid algorithm on the other hand can tackle with any $\Phi_p$-optimality regarding any smooth function of the parameters. The time efficiency of the latter has been shown in Table \ref{table1}. Now we use it to investigate the influence of different design criteria on the choice of samples. Here, we focus on the A and D criteria. For fixed $N=10,000$ and various choices of $n$, we derive both the A-optimal and D-optimal samples and evaluate their efficiencies against each other under the opposite criteria. As shown by Table \ref{table2}, an A-optimal sample does not necessarily perform well under the D-criterion and vice versa. The discrepancy between the two criteria is larger when the sample size to population size ratio, i.e. $n/N$, is small. That is because less common points will be shared by the two samples when there are relatively fewer points to be selected.

\begin{table}[h]

\begin{center}
\begin{tabular}{|c|c|c|c|c|}
\hline
$n$&500&1000&3000&5000\\ \hline
A-efficiency & 0.6511&0.7064&0.8485&0.9439\\ \hline
D-effficiency&0.8022& 0.8287&0.8951&0.9484\\ \hline
\end{tabular}
\end{center}
\caption{The A-efficiency of the D-optimal sample and the D-efficiency of the A-optimal sample. The population is of fixed size 10,000. Each point consists of the constant term and a random vector from $10$ dimensional standard normal distribution.}\label{table2}
\label{default}
\end{table}

Let $\lambda_1,...\lambda_k$ be the eigenvalues of the inverse of the information matrix $M^{-1}$. Recall that A-criterion aims to minimize $\sum^k_{i=1}\lambda_i$ and D-criterion aims to minimize $\Pi^k_{i=1}\lambda_i$. As pointed out by \citet{ouyang2017batch}, the A-criterion treats the eigenvalues equally important while the D-criterion allows certain eigenvalues to be large as long as some eigenvalues are small enough. We can understand this by looking at the confidence region of the parameters, which is in the form of $\{\beta: (\beta-\hat{\beta})^TM(\beta-\hat{\beta})\leq c\}$ for some constant $c$. It represents an ellipsoid, where the lengths of the axes are proportional to the $\lambda_i$'s. In order for the volume of the ellipsoid to be small, it is preferable to have a long but narrow shaped ellipsoid rather than a round shaped ellipsoid. The former could be achieved by choosing certain eigenvalues to be small enough while keeping others within some range.

\subsection{Ordinal response: cumulative link model}
Recently, \citet{yang2017d} studied the optimal designs for the cumulative link model
\begin{eqnarray}\label{eqn:723}
h(\gamma_j)=\theta_j-z^T\beta, &&j=1,2,...,J-1,
\end{eqnarray}
where $\gamma_j=P(Y\leq j)$ is the probability for the ordinal response $Y\in \{1,2,...,J\}$ to be no larger than $j$ and $h$ is the link function. They worked on the classical design problem where the weights, $w_i$'s, are not bounded by $\epsilon$. When $h$ is taken to be the logit link function, (\ref{eqn:723}) is reduced to the proportional odds model. \citet{perevozskaya2003optimal} studied optimal designs for this model when $z$ is univariate. For the general model (\ref{eqn:723}) with $z\in \R^d$, \citet{yang2017d} revealed an interesting fact that the minimum number of experimental settings is $d+1$, which is strictly less than the number of parameters $d+J-1$. This fact is of more interest for a classical design problem than the sampling problem here, since we are enforcing the constraint $w\leq \epsilon$ for the weights of the designs points and a large number of supporting points is desired. By direct calculations, the information matrix for the parameter $\theta=(\beta^T,\theta_1,\theta_2,...,\theta_{J-1})^T$ at the covariate $z$ is
\begin{eqnarray*}
M_z&=&\sum^J_{j=1}\frac{1}{\pi_j}x_jx_j^T,\\
x_j&=&
\left(\begin{tabular}{c}
$[f_j(z)-f_{j-1}(z)]z$\\
$f_j(z)e_j-f_{j-1}(z)e_{j-1}$
\end{tabular}\right),
\end{eqnarray*}
where $\pi_j=P(Y=j)$, $f_j(z)=(h^{-1})'(\theta_j-z^T\beta)$ and $e_j$ is the $J-1$ dimensional vector with the $j$-th elements as $1$ and other elements as $0$. Here, we adopt the conventions $f_0(z)=f_J(z)=0$ and $e_0=e_J=0$. For the proportional odds model, we have $f_j(z)=exp(\theta_j-z^T\beta)/(1+exp(\theta_j-z^T\beta))^2$. This representation of the information matrix is much simpler than that in Yang, Tong and Mandal (2015) and is more suitable for the application of the algorithms proposed in this paper. It also provides a mathematical explanation for why the minimum number of designs points could be $d+1$ instead of $d+J-1$. Note that the matrix $M_z$ based on a single point could be of rank larger than $1$, while the rank can only be $1$ for traditional models such as the example in Section \ref{sec:simple}. As a result, it is possible to collect $d+1$ points so that the summation of the information matrices based on them is of full rank.

Now we apply the algorithm to the wine quality data set, which is publicly available in the UCI Machine Learning Repository (https://archive.ics.uci.edu/ml/datasets/Wine+Quality). The data were collected to build a model regarding how the quality of the wine is affected by various possible factors. Hence we shall be more interested in the parameter $\beta$ in (\ref{eqn:723}) rather than the $\theta_j$'s. The data includes $11$ independent variables, which are all chemical components with their levels easily measured by devices. They are all continuous variables, and we shall standardize by their means and standard deviations so that the design criteria could be validly adopted in the sampling process. The quality of each wine will be rated by a wine expert by the integer scale from $0$ to $10$. There are $N=4898$ wines in the data and Table \ref{table6} below provides the frequency count of the quality ratings of wines. There are $7$ categories of response observed in the full data and the categories of $3$ and $9$ are very rare. As a result, there is a chance that some categories may not be observed in the sample. This is not an issue since we will only be interested in $\beta$.

\begin{table}[h]
\begin{center}
\begin{tabular}{|c|c|c|c|c|c|c|c|}
\hline
Quality  &  3 &   4 &   5  &  6 &   7  &  8 &   9\\\hline
Frequency   &  20  &163 &1457& 2198 & 880 & 175 &   5\\\hline
\end{tabular}
\end{center}
\caption{The frequency count of the quality ratings of the wines.}\label{table6}
\end{table}

We use this data set only to illustrate our method instead of making any scientific discoveries, hence we will pretend the response to be unobserved before sampling. To implement the algorithm, we need to calculate the information matrix of each single candidate point, which relies on the values of the parameters. It is natural to propose a two stage procedure of sampling. The first stage consists of a simple random sample based on which the parameters will be estimated. The second stage tries to select the rest number of points so that the combined sample is the most informative. Suppose MSE is the primary target, we shall choose A-criterion as the measurement of information. Another issue is the size allocation of the samples from the two stages. If the percentage of the first stage sample, denoted by $r$, is too large, the estimate of the parameters will be reliable, however, there is less room to improve the quality of the sample on the second stage. On the other hand, if $r$ is too small, the estimate of the parameters is subject to large variability which induces further doubts on the determination of the second stage sample. There is a trade off in the choice of $r$ for balancing these two issues.

Since the data generation mechanism is unknown, we shall evaluate the performance of the sampling methods through bootstrapping, where a big number, say $B$, of datasets are generated by resampling from the original data with replacement. A model fitted to the original data is treated as the true model and the algorithm will be applied to to each resampled data to derive the corresponding sample. The performance of a sampling method is evaluated by the total MSE of all components of $\beta$ as defined by the quantity $B^{-1}\sum^B_{b=1}||\hat{\beta}^b-\hat{\beta}||^2$, where $\hat{\beta}$ is the estimate of $\beta$ from the original data and $\hat{\beta}^b$ is the estimate of $\beta$ based on the sample derived from the $b$-th bootstrapped data. By a regular data analysis procedure, we find that model (\ref{eqn:723}) is well fitted to the original data with $11$ main effects and $17$ two way interactions as the regressors. That means $\beta$ has $28$ components. Further, we set the combined sample size to be $n=500$ and the number of iterations to be $B=1000$.

Figure \ref{fig:MSE1} gives the comparison between the derived sample and the random sample in terms of the ratio of the total MSE as defined earlier when $r$ takes the values of $0.2,0.3,...,0.9$. The smaller the ratio, the more advantage our designed sample is as compared to the random sample. The value $r=0.1$ is not included since the ratio is slightly larger than $1$ and hence excluded from consideration. The smallest ratio as $0.4570$ is obtained when $r=0.4$. We also calculate the ratio of MSE for each component of $\beta$ when $r=0.4$ and draw a histogram of these values in Figure \ref{fig:hist}. It can been seen that only $3$ out of $28$ components have the ratio larger than one. The rest is smaller than one with the majority below $0.6$.

\begin{figure}[H]
\centering
\includegraphics[scale=0.3]{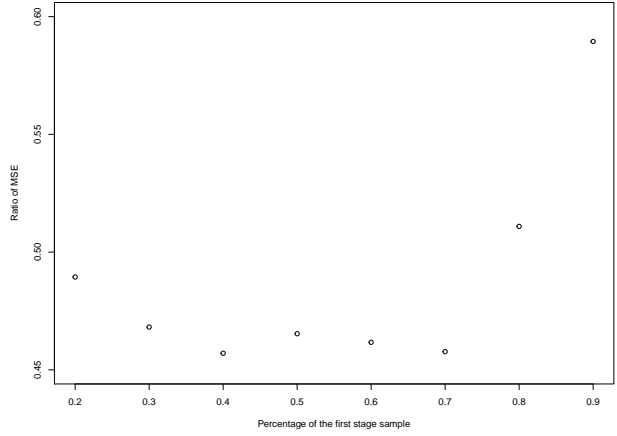}
\caption{The ratio of total MSE of all parameters between the two stage designed sample and the random sample. Here the sample size is $500$ and the data has $4898$ observations. The lowest ratio is $0.4570$ when the percentage of the first stage sample is $40\%$.}\label{fig:MSE1}
\end{figure}

\begin{figure}[H]
\centering
\includegraphics[scale=0.3]{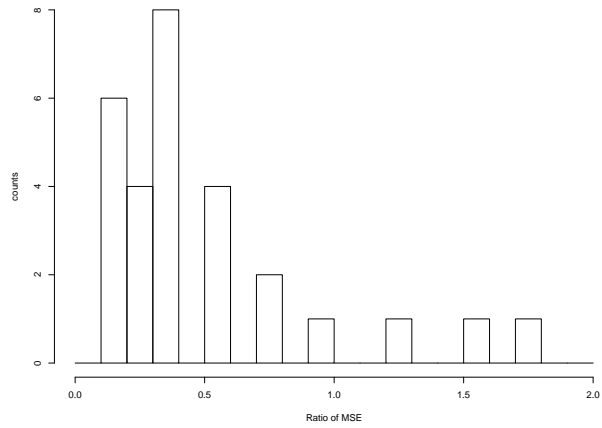}
\caption{The histogram of ratios of MSE for all parameters between the two stage designed sample and the random sample. The sample size is $500$ and the first stage sample in the two stage designed sample occupies $40\%$ of the total sample, namely of size $200$.}\label{fig:hist}
\end{figure}

Similar to the analysis of wine quality data, we apply another data set from the US National Health and Nutrition Examination Survey (NHANES) between 2007 and 2014 \citep{patel2016database}, a nationally representative, cross-sectional biannual health survey, which included health questionnaire, laboratory (i.e. urinary mercury) and clinical data. The gout as our target event, was collected by self-reported questionnaire: "Has a doctor or other health professional ever told you that you had gout?" Answering yes to this question was identified positive for Gout. We remove the factors that differentiated among individuals and those that have a majority (90\%) of the observations below a detection limit threshold as defined in the NHANES codebook. Also, we discard the factors with missing values and the most important factors are determined using random forest and known experience from clinical professionals. As a result, there are a total of 8 factors with 16923 participants.    

The factors are continuous, and standardized by their means and standard deviations, which are uric acid (umol/L), red cell distribution width (\%), mean cell volume (fL), alanine aminotransferase ALT (IU/L), age in years of the participant at the time of screening, waist circumference (cm), blood pressure (mm Hg), and blood urea nitrogen (mg/dL). We follow the two-stage process, where the first stage is a simple random sampling the original subset data by $r$ between 0.3 and 0.9 to estimate the parameters in a logistic model. The second stage is to select the remaining samples using hybrid algorithm for D-criterion. The samples selected from two stages form to the different sizes, in which we choose the sample sizes with 1500, 2250, 3000, 3750 and 4500 in total. After 1000 iteration bootstrapping, we obtain the ratio of total MSE of all parameters between the two stage designed sample and the random sample sizes from 1500 to 4500, respectively (Figure \ref{fig:MSE2}). 

\begin{figure}[H]
\centering
\includegraphics[scale=0.5]{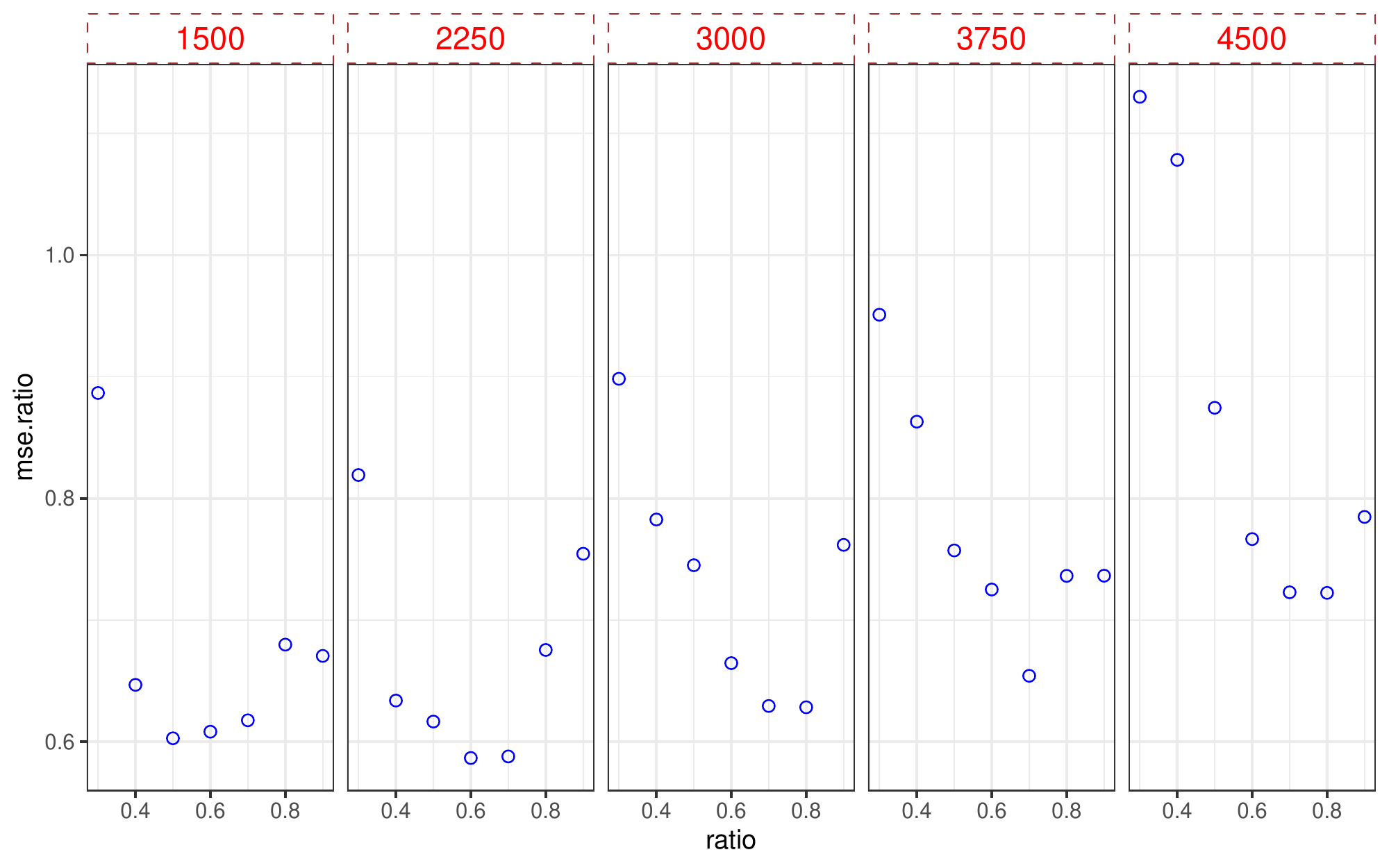}
\caption{The ratio of total MSE of all parameters between the two stage designed sample and the random sample. Here the majority of ratio is less than 1, except that $r$ takes values of 0.3 and 0.4 with total sample size 4500.}\label{fig:MSE2}
\end{figure}
 
Figure \ref{fig:MSE3} shows the histogram of the ratio of MSE between the two stage designed sample and the random sample for each component in our logistic model. 

\begin{figure}[H]
\centering
\includegraphics[scale=0.5]{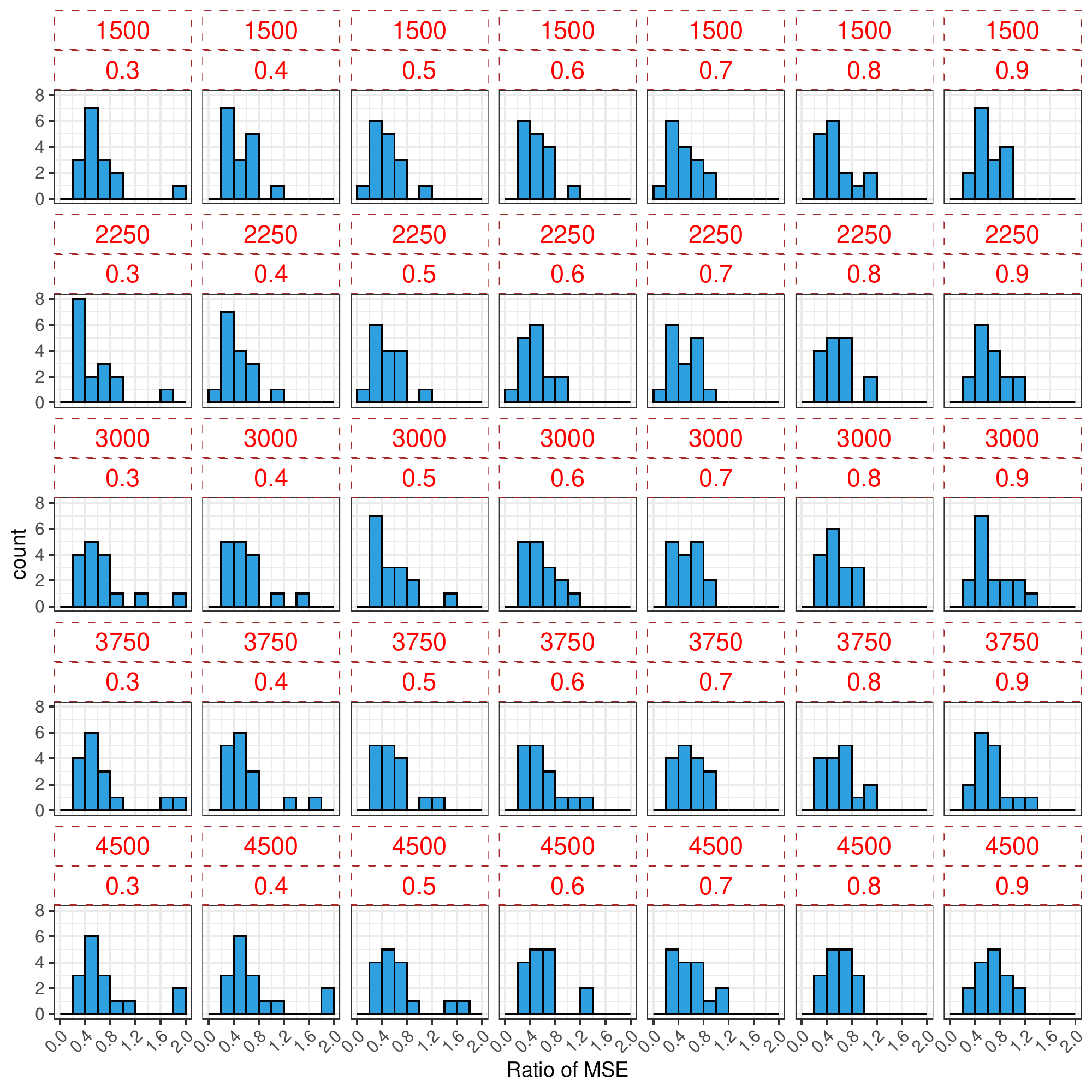}
\caption{The ratio of total MSE of each parameters between the two stage designed sample and the random sample. Here there is 8.57\% component whose the ratio of total MSE is larger than 1.}\label{fig:MSE3}
\end{figure}
          
To distinguish the selection bias, we calculate the ratio of MSE for each component including only variance. Figure \ref{fig:MSE4} represents the histogram of the ratio of MSE between the two stage designed sample and the random sample for each component incorporating variance. There is only the intercept whose the ratio of MSE for variance is larger than 1. 

\begin{figure}[H]
\centering
\includegraphics[scale=0.5]{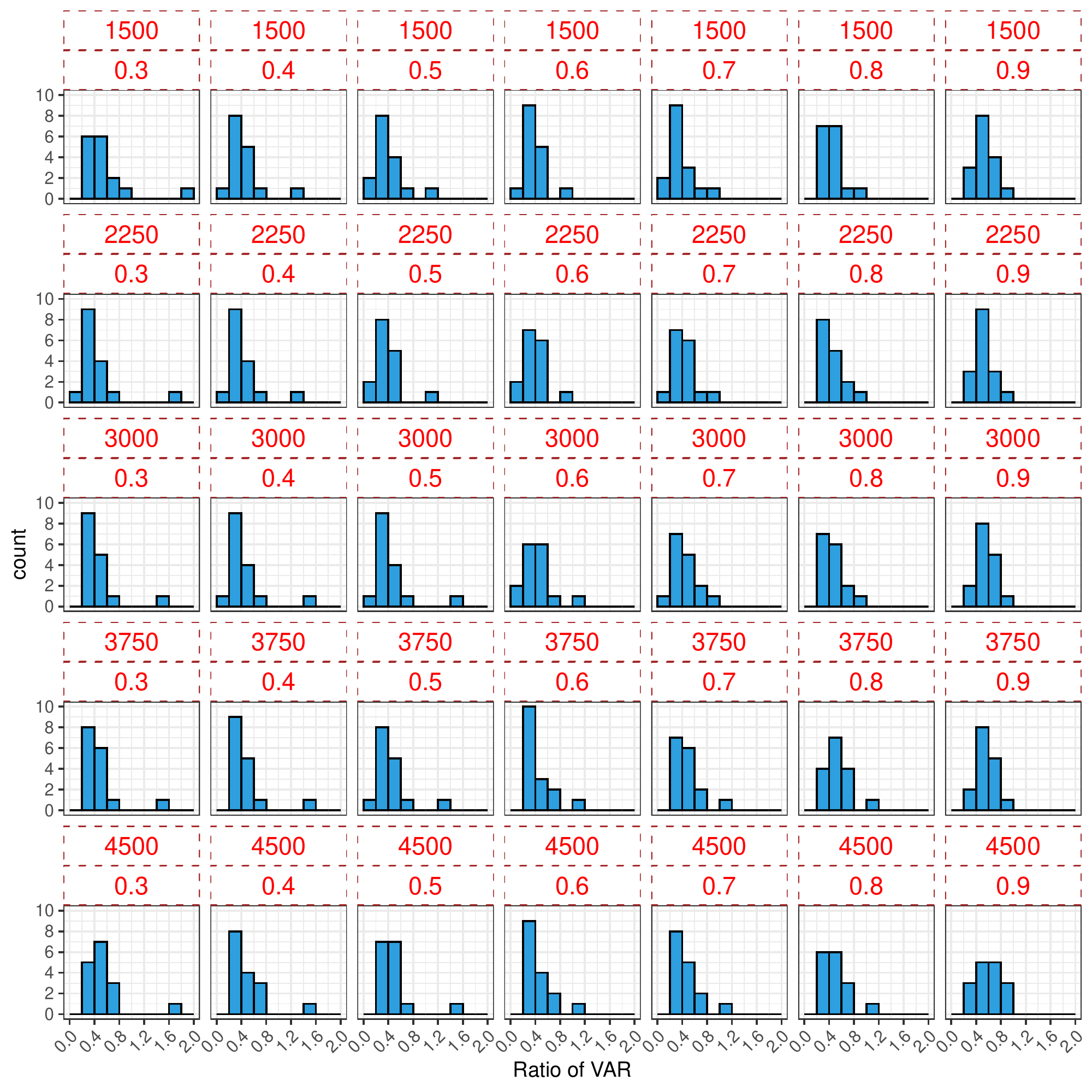}
\caption{The ratio of MSE of each parameters between the two stage designed sample and the random sample. Here there is 3.9\% component whose the ratio of MSE for variance is larger than 1.}\label{fig:MSE4}
\end{figure}


%
%
%


\bibliographystyle{plainnat} 
\bibliography{mybibtex} 


\end{document}